\documentclass[12pt]{amsart}
\usepackage{lineno}
\usepackage{graphicx,color} 
\usepackage{cite}
\usepackage{geometry}

\newcommand\R{\mathbb{R}}
\newcommand\C{\mathbb{C}}
\newcommand\be{\begin{equation}}
\newcommand\ee{\end{equation}}
\newcommand\beq{\begin{equation}}
\newcommand\eeq{\end{equation}}
\newcommand\bed{\begin{definition}}
\newcommand\ed{\end{definition}}
\newcommand\brem{\begin{remark}}
\newcommand\erem{\end{remark}}
\newcommand\benum{\begin{enumerate}}
\newcommand\enum{\end{enumerate}}
\newcommand\beit{\begin{itemize}}
\newcommand\eit{\end{itemize}}
\newcommand\bt{\begin{theorem}}
\newcommand\et{\end{theorem}}
\numberwithin{equation}{section}

\newtheorem{theorem}{Theorem}[section]
\newtheorem{assumption}{Assumption}[section]
\newtheorem{lemma}{Lemma}[section]

\makeatletter
\newcommand{\eqnum}{\refstepcounter{equation}\textup{\tagform@{\theequation}}}
\makeatother
\newtheorem{definition}{Definition}[section]

\newtheorem{remark}{Remark}[section]

\usepackage{lineno,hyperref}
\usepackage{amsmath}
\usepackage{subcaption}

\begin{document}

\title[Quantum open books]{On open book analogs of quantum graphs}

\author{Setenay Akduman}
\address{Department of Mathematics, Izmir Democracy University, T\"{U}RK\.{I}YE} 
\email{setenay.akduman@idu.edu.tr}
\author{Peter Kuchment}
\address{Department of Mathematics, Texas A\& M University, USA}
\email{kuchment@tamu.edu}

\date{\today}
\keywords{Open book structure, quantum graph, Laplace operator}
\subjclass{34B45, 34L40, 34B10, 81Q35}
\maketitle

\begin{abstract}
Quantum graphs have become in this century a favorite playground for mathematicians, mathematical physicists, and chemists, due to their manifold applications as models of thin structures, as well as presenting sometimes simpler playground for hard higher dimensional problems.

It was clear from some applications that thin surface structures (looking as stratified varieties) also arise, for instance in photonic crystals theory and dynamical systems. However, both justification and studying of these models is much harder and very little progress has been made by now. The goal of this note is to set down some basic notions and results for such structures.

The name ``open book'' has been used for such geometric structures in topology and comes from an image of several smooth $n$-dimensional ``pages'' bound to an $(n-1)$- dimensional ``binding.''
\end{abstract}

\section{Introduction}\label{Intro}

The so called metric (or quantum, if equipped with a ``Schr\"odinger''operator) graph~\cite{BerKuc} is just a graph, each edge of which is endowed with a metric that identifies it with an interval of the real line $\R$. In somewhat more pretentious way, metric graphs can be considered as one-dimensional CW complexes with Riemannian metric (or stratified 1D Riemannian varieties). Among the origins of these objects are in particular attempts to approximate very ``thin'' and branching structures (being them circuits of nanometers-thick ``quantum wires,'' or photonic structures, thin waveguides, etc.) by graphs (see, e.g. \cite{AGA,KuPBG,FK4,FK5,BerKuc,RS2,SmilSol,Post} and references therein). Probably the first usages of such models were seen in chemistry \cite{Pauling,RuedS}. Another source was studying dynamical systems and random processes that exhibit slow motion along a 1D structure and fast motion across it. Averaging the fast motions produces a process on a graph \cite{Fr,FW}. Such object have also been used as ``toy'' models for hard problems of quantum chaos (see, e.g. \cite{BerKeatSmi,Smil07,Smil13,KoSm03,KoSm99}). All this has led to a booming area of analysis on quantum graphs.

One cannot help to notice that higher (especially two-) dimensional singular varieties equipped with ``differential equations'' arise as naturally as the quantum graphs. This happens, for instance, in the already mentioned averaging \cite{FW,FWopb}, in the limits of thin ``almost two-dimensional'' structures \cite{CorbinKuch}. And the use of graph models for photonic crystals in \cite{KuPBG,FK5} was somewhat a cheating, since only 2D cylindrical structures and an appropriate polarization were used to reduce the study to a graph cross-section. 
The photonic band-gap structures often look more like thin 2D varieties of dielectric surrounded by air. However, even justifying such surface models as limits of the ``thin'' ones turns out to be a much harder task than in the quantum graph case, with the only limited results known to the authors being those in \cite{FWopb,CorbinKuch}. Studying the models themselves has apparently also been, to the best of our knowledge,  non-existent.

Notice that, as in the quantum graph case, these surface structures in all interesting cases have singularities. Thus, one would want to consider a $n$-dimensional stratified variety equipped with a Riemannian metric. This would allow singularities of arbitrary co-dimension. However, this is a technically much more difficult task, and to the best of our knowledge, so far no development has happened in this direction. Allowing higher co-dimension singularities in analogs of the results of \cite{CorbinKuch} would provide very useful models for the further study. In this paper (as in \cite{CorbinKuch}) we consider only singularity strata that are smooth manifolds of co-dimension one in the variety.

The basic notions concerning open book structures are introduced in Section \ref{S:OpenBooks}, while the operators of interest and quantum open books are described in Section \ref{Operators}. The next Section \ref{S:selfadjoint} contains discussion of the conditions of self-adjointness of the main operator.  The last Section \ref{Comments} is dedicated to some further comments.

\section{Open book structures}\label{S:OpenBooks}

Open book structures have been used in geometry and topology, e.g. in the higher dimensional knot theory (see\cite{Ranicki,WinkelRan,Winkel}).  We adopt a simplistic definition of an open book as a \textbf{compact stratified $n$-dimensional Riemannian variety, which justifies the name ``open book,'' with smooth $n$-dimensional pages meeting at smooth $(n-1)$-dimensional ``bindings.''} Thus, a \begin{underline}local\end{underline} picture near a binding point looks as in Fig. \ref{F:pages}, which justifies the name ``open book.''
\begin{figure}[ht!]
  \centering
  \includegraphics[width=0.4\textwidth]{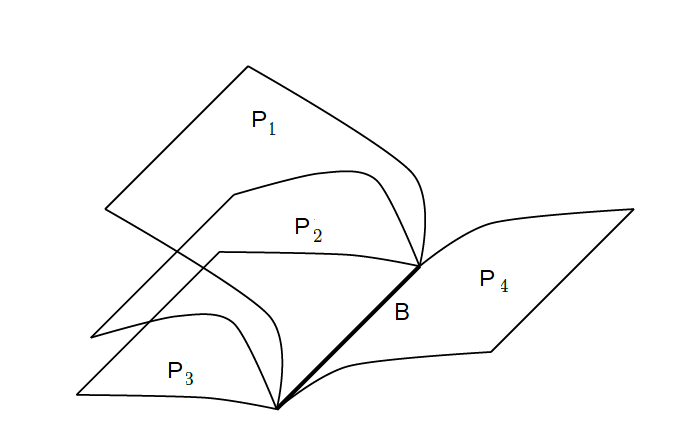}
  \caption{The local structure, explaining the name ``open book.}\label{F:pages}
\end{figure}
Figure \ref{F:3spheres} provides an example of such a variety consisting of pieces of spheres (pages) meeting at a circular boundary (binding).

\begin{definition}
An \textbf{compact metric open book variety} $M$ consists of finitely many ``pages'' $P_1,\dots,P_m$, which are smooth compact $n$-dimensional Riemannian manifolds, the boundary of each of them consisting of one or more from a finite list of  $(n-1)$-dimensional smooth ``bindings'' $B_1,\dots,B_l$. 

We will use the natural terms ``\textbf{a page adjacent to a binding}'' and `\textbf{a binding adjacent to a page},'' if the said binding belongs to the boundary of the said page.

We also denote by $k_j$ the number of pages adjacent to the binding $B_j$.
\end{definition}
For instance, the variety $M$ shown in Fig. \ref{F:3spheres} consists of one circular binding with six spherical caps as pages. One can note that globally, the ``pages,'' as well as the ``bindings'' can contain components of different topology, as in the ``dumbbell'' shown in Fig. \ref{F:dumbell}, consisting of spherical cap pages and a cylindrical one. 

\brem
\begin{figure}[ht!]
  \centering
  \includegraphics[width=0.4\textwidth]{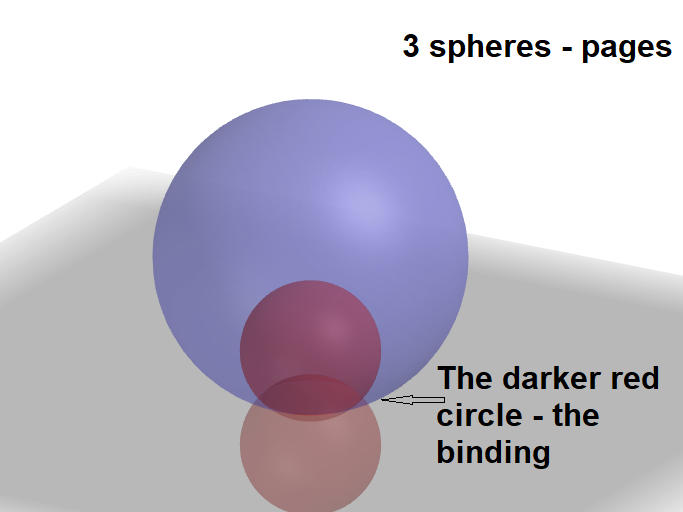}
  \caption{Three spherical ``pages'' (in fact, six pieces of them) connected at a circular ``binding.''}\label{F:3spheres}
\end{figure}

 \begin{figure}[ht!]
  \centering
  \includegraphics[width=0.4\textwidth]{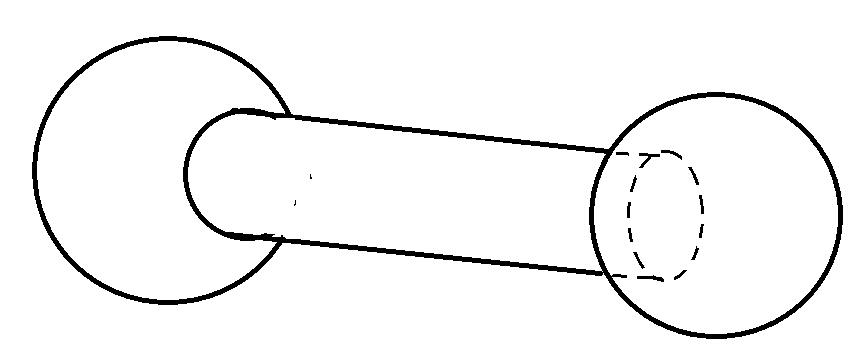}
  \caption{A ``dumbbell'' open book structure with two circular bindings, one cylindrical, and two spherical cap pages.}\label{F:dumbell}
\end{figure}
Notice that lower dimensional (i.e., of dimension $(n-2)$ and less) smooth strata are not allowed. The reason is that, as in \cite{CorbinKuch}, we have not been able to deal with this much more technically challenging case. One can see that small perturbations of centers/radii of spheres in Fig. \ref{F:3spheres} would produce such zero-dimensional strata.

We also note that our definition of an open book structure is somewhat less restrictive than the one in topology\footnote{
One can notice a (superficial) resemblance to the Bruhat-Tits buildings \cite{Build_ronan,Bruhat-Tits,Build_Remy}.} (see, e.g., \cite{Winkel}).
\erem

\section{Operators on open books. ``Quantum Open Books''}\label{Operators}

\begin{definition}Let $M$ be a compact metric open book variety.
\begin{itemize}
\item The variety $M$ is equipped with the natural measure arising from the Riemannian metric on each page. This allows one to introduce the space 
    \be
    L_2(M):=\bigoplus\limits_i L_2(P_i).
    \ee
     of square integrable functions on $M$.
\item	We denote by $\widetilde{H}^2(M)$ the space 
\be \widetilde{H}^2(M):=\bigoplus\limits_i H^2(P_i)
\ee
 consisting of the functions $u$  on $M$ that on each page $P_i$ belong to the Sobolev space $H^2(P_i)$ with the norm
	$$||u||^2_{\widetilde{H}^2(M)}:=\sum\limits_{i}||u||^2_{H^2(P_i)}.$$
{\bf Notice that no compatibility at the bindings is required}.
\item Analogously, Sobolev function spaces can be defined on each binding $B$.
\item A {\bf quantum open book structure} is a metric open book  $M$ equipped with a differential operator $\mathcal H$ on pages and ``appropriate'' junction conditions across each binding. 
\end{itemize}
\end{definition}

Although one can consider more general elliptic operators $\mathcal H$, in this work we will assume that it acts as the Laplace-Beltrami operator $-\Delta_P$ on each page $P$ (the negative sign is used, as is common, to make the operator positive) : 
\be 
\mathcal H=\bigoplus\limits_P (-\Delta_P),
\ee
with the domain $\widetilde{H}^2(M)$. In what follows, we will drop the subscript $P$ in $-\Delta_P$, whenever this does not lead to confusion.

The following statement is obvious:
\begin{lemma}\label{L:bounded}
The operator $\mathcal H$ is bounded as an operator from the space $\widetilde{H}^2(M)$ to $L_2(M)$.  
\end{lemma}
It is clear that this operator has a huge kernel, unless appropriate (elliptic) \textbf{boundary} (or rather \textbf{junction}) \textbf{conditions} are imposed at each binding.

\subsection{Binding (junction) conditions}\label{SS:junction}
At each binding $B$, one can define the ``normal derivative'' operator $\dfrac{\partial}{\partial \nu}$ that produces a vector of normal to $B$ derivatives of functions from all the adjacent pages. Namely, if $B$ is adjacent to pages $P_1,\dots,P_k$, then having a function $u$ on $M$ that belongs to $\widetilde{H}^2(M)$, one can define the $n$-vector of functions with components in $H^{1/2}(B)$ as follows:
\be
\frac{\partial u}{\partial \nu}:=\left\{\frac{\partial u_{P_1}}{\partial \nu}, \dots, \frac{\partial u_{P_k}}{\partial \nu}\right\}.
\ee
We denote by $u|_B$ the $n$-vector of restrictions of the function $u$ from all adjacent pages: 
\be
u|_B:=\{(u|_{P_1})|_B,\dots, (u|_{P_k})|_B\}.
\ee
Let us now have smooth $k\times k$ matrix functions $A(x), C(x)$ on $B$ and define the desired \textbf{junction (= binding) conditions} as follows\footnote{We skipped from the letter $A$ to $C$ in these matrix notations, not to mix them up with a binding $B$.}:
\be\label{E:junction}
A(x)u(x)|_B +C(x)\frac{\partial u(x)}{\partial \nu}=0 \mbox{ for all } x\in B.
\ee
\begin{assumption} \label{A:bv}
We will assume that these smooth matrix functions are defined on each binding (note that the size of these matrices is equal to the number $k$ of pages adjacent to the binding, and thus can be different for different bindings in the same open book).
\end{assumption}

\begin{definition}\label{D:S-L}
We impose the following \textbf{ellipticity condition}:

For any binding $B$, the corresponding matrices $A$ and $C$, any $x\in B$, and any $\lambda>0$ the following condition holds:
\be\label{E:ellipt}
\det (A(x)-\lambda C(x))\neq 0.
\ee
\end{definition}
\begin{remark}\label{R:maxrank}
Notice, that this implies in particular that the $k\times 2k$ matrix $(A,C)$ has (the maximal possible) rank $k$.
\end{remark}
We now incorporate the conditions into defining our final choice of the space.
\begin{definition}\label{D:finalspace}
We denote by $H^2_0(M)$ (or, in more detailed notations $H^2_{A,C}(M)$) as the subspace of 
$\widetilde{H}^2(M)$ that consists of functions satisfying the binding conditions (\ref{E:junction}) for all bindings $B$ and any $x\in B$.
\end{definition}
We now define our final operator of interest as the restriction of the original $\mathcal{H}$ to the just defined subspace:
\begin{definition}
    \be
    \mathcal{H}_0:=\mathcal{H}|_{H^2_0(M)}.
    \ee
\end{definition}
Our first result is the following:
\begin{theorem}
The operator
\be
\mathcal{H}_0: H^2_0(M)\mapsto L_2(M)
\ee
is bounded and Fredholm\footnote{The Shapiro-Lopatinskyi conditions are sufficient for Fredholmity, but might not be necessary. Thus, a more general (although implicit) condition will be used in Section \ref{Spectrum}.}.
\end{theorem}
\begin{proof}
As a restriction of a bounded operator, the operator is bounded.  

The Fredholm property follows by usual techniques from ellipticity, which can be checked locally. 
Indeed, inside the pages the operator Laplace-Beltrami is clearly elliptic. 
Near a point $x_0$ of a binding $B$, one checks that the condition (\ref{E:ellipt}) provides the so called \textbf{covering}, or \textbf{Shapiro-Lopatinskyi conditions}, which imply the ellipticity near $x_0$. The arising local parametrices can be combined in the usual way using partitions of unity, which provides the Fredholm property of the boundary value problem operator (see, e.g. \cite{Agran,LionsMagenes}).
\end{proof}

\section{Selfadjointness}\label{S:selfadjoint}
Here we are interested in establishing conditions of self-adjointness of the quantum open book Hamiltonian $\mathcal{H}_0$.
\begin{theorem}\label{self-adjointness}
The following conditions are equivalent:
 \begin{enumerate}
 \item Operator $\mathcal H_0$ is self-adjoint as an unbounded operator in $L_2(M)$ with the domain $H^2_0(M)$.
 \item The  matrix $A(x)C^*(x)$ is self-adjoint for any $x\in B$ and arbitrary binding $B$. 
\item For each binding $B$ that has $k$ pages attached, there exists a unitary $k\times k$ smooth matrix function $\mathcal{U}(x)$ such that the condition (\ref{E:junction}) is equivalent to the following one:
\begin{equation}\label{condition_b}
i(\mathcal{U}(x)-\mathcal{I})u|_B(x)+(\mathcal{U}(x)+\mathcal{I})\frac{\partial u(x)}{\partial \nu}|_B=0,
\end{equation}
where $\mathcal{I}$ is the $k\times k$ identity matrix.
\end{enumerate}
 \end{theorem}

\begin{proof} We will prove this in the following order: 
$$
 (1) \Rightarrow (2) \Rightarrow (3)  \Rightarrow (2) \Rightarrow (1). $$ 
\begin{itemize}

\item  $(1) \Rightarrow (2).$ Let us assume that the operator $\mathcal{H}_0$ is self-adjoint. We need to show that the matrix $A(x)C^*(x)$ is self-adjoint for any $x\in B$. 
    
    Choose a function $u \in D(\mathcal{H}_0)=H^2_0(M)$. Let also $B$ be a binding and $x_0\in B$. Consider the expression 
\be
<\mathcal{H}_0u,v>-<u,\mathcal{H}_0v>,
\ee
where the brackets denote the hermitian inner product in $L_2(M)$. After integration by parts over the pages, this boils down, as usual,  to the expression
\be\label{E:parts}
\sum_i \int\limits_{B_i} (\dfrac{\partial u}{\partial \nu}|_B\cdot v|_B -u|_B\cdot \dfrac{\partial v}{\partial \nu}|_B),
\ee
where $a\cdot b=\sum a_j\overline{b_j}$ is the sesqui-linear inner product in $\C^n$.

Self-adjointness means that if for some $v$ this expression vanishes for all $u$ in the domain of $\mathcal{H}_0$, then $v$ should be also in the domain and thus also satisfy the junction condition (\ref{E:junction}).

To prove that this is indeed true, we choose an arbitrary function $h(x)$ on $M$ supported near $x_0\in B$ and such that it is smooth up to a boundary on each page adjacent to $B$ (no conditions at the binding are required). The values on $B$ of a function $v\in \widetilde{H}^2(M)$ and of its normal derivative $\dfrac{\partial v}{\partial \nu}$ can be chosen independently, so we can choose such a function that 
\be
v(x)|_B=-C(x)^*h(x)
\ee 
and 
\be
\frac{\partial v (x)}{\partial \nu}|_B=A(x)^*h(x).
\ee
 Plugging these $u$ and $v$ into (\ref{E:parts}), integrating by parts, and using (\ref{E:junction}), one concludes that the expression in  (\ref{E:parts}) vanishes. This means that $v$ is in the domain of $\mathcal{H}_0^*$, and due to self-adjointness it also belongs to the domain of  $\mathcal{H}_0$. Hence, it satisfies (\ref{E:junction}), which implies the following:
\be
\int_B (-AC^*+CA^*)h=0.
\ee
Since the function $h$ on $B$ was arbitrary, we conclude that $AC^*$ is self-adjoint.

\item $(2) \Rightarrow (3)
$. We start with the following auxiliary result.
\begin{lemma}\label{L:unitary}
	Let matrices $A(x)$ and $C(x)$ be such that 
 \begin{enumerate}
     \item the $k\times 2k$ matrix $(A,C)$ has the (maximal possible) rank $k$
     and 
     \item the matrix $AC^*$ is self-adjoint. 
 \end{enumerate}
 Then, for any real $z\neq 0$ and $x\in B$, the matrix $A(x)+izC(x)$ is invertible and the matrix 
 \be
 \sigma(z,x):=-({A(x)+izC(x)})^{-1}(A(x)-izC(x))
 \ee 
 is unitary.
\end{lemma}
This lemma is proven in \cite[Lemma 1.4.7]{BerKuc} for the case of constant (i.e., independent on $x$) matrices $A$ and $C$.
The corresponding proof applies without any change to our case.

Let us return to the proof of the implication $(2) \Rightarrow (3)$.

So, let $A(x)C^*(x)$ be self-adjoint for any $x\in B$. Then matrices $A(x)$ and $C(x)$ satisfy the assumptions of Lemma \ref{L:unitary}, and hence the matrix 
\be
-2i({A(x)+izC(x)})
\ee
 is invertible.  Then we get
\be
\begin{array}{l}
-2i({A(x)+izC(x)})^{-1}A(x)=\\ -2i({A(x)+izC(x)})^{-1}\cdot\frac{1}{2}\;((A(x)-izC(x))+(A(x)+izC(x))\\= i(\sigma(z,x)-\mathcal{I}).
\end{array}
\ee
Similarly,
\be
\begin{array}{l}
-2i({A(x)+izC(x)})^{-1}C(x)=\\ -2i({A(x)+izC(x)})^{-1}\cdot\frac{1}{2iz}((A(x)+izC(x))-(A(x)-izC(x))\\= -\frac{1}{z}(\mathcal{I}+\sigma(z,x)).
\end{array}
\ee
Using these formulas, after multiplying equation (\ref{E:junction}) on the left by the matrix $-2i({A(x)+izC(x)})^{-1}$, we get for all real $z\neq 0$
$$ i(\sigma(z,x)-\mathcal{I})u_C(x)-\frac{1}{z}(\mathcal{I}+\sigma(z,x))\dfrac{\partial u_C(x)}{\partial\nu}=0,$$

Setting $z=-1$ and denoting 
$$\mathcal{U}(x):=\sigma(-1,x)={\sigma(1,x)}^{-1},$$
 we obtain the following condition
$$i(\mathcal{U}(x)-\mathcal{I})u_C(x)+(\mathcal{U}(x)+\mathcal{I})\dfrac{\partial u_C(x)}{\partial\nu}=0.$$
It is exactly of the form  (\ref{condition_b}).
\item The implication $(3) \Rightarrow (2)$ follows by a direct computation. Indeed, since   
\be
    A=i(\mathcal{U}(x)-\mathcal{I}) \mbox{ and }C=(\mathcal{U}(x)+\mathcal{I}),
\ee 
    we have
\be
\begin{array}{l}
A(x)C^*(x)= i(\mathcal{U}(x)-\mathcal{I})(\mathcal{U}^*(x)+\mathcal{I})\\ \\= i(\mathcal{U}(x)-\mathcal{I})(\mathcal{U}^{-1}(x)+\mathcal{I})= i(\mathcal{U}(x)-\mathcal{U}^{-1}(x))
\end{array}
\ee
and
\be
\begin{array}{l}
(A(x)C^*(x))^*= C(x)A^*(x)= (\mathcal{U}(x)+\mathcal{I})(\mathcal{U}^*(x)-\mathcal{I})(-i)\\
\\
= i(\mathcal{U}(x)+\mathcal{I})(\mathcal{I}-\mathcal{U}^*(x))= i(\mathcal{U}(x)-\mathcal{U}^{-1}(x)).
\end{array}
\ee

\item We skip the proof of the implication $(2) \Rightarrow (1)$, since it follows by reversal of the construction used to prove converse statement $(1) \Rightarrow (2)$.

\end{itemize}
\end{proof}
\begin{remark}\label{R:non-uniqueness}\indent
\begin{enumerate}
\item 	Condition (\ref{condition_b}) is a special case of the condition of type (\ref{E:junction}). The latter is easier to check, but its drawback is that the choice of $A(x)$ and $C(x)$ is not unique. For instance, multiplying the two matrices on the left by an invertible matrix(-function) provides an equivalent condition. For instance, for any $z\neq 0$, the following matrices
$$A(z,x)=i(\sigma(z,x)-\mathcal{I})$$
and
$$C(z,x)=-\frac{1}{z}(\sigma(z,x)+\mathcal{I})$$
satisfy the maximal rank condition, as well as  self-adjointness of $AC^*$, and provide equivalent binding conditions. 

\item On the other hand, the matrix $\mathcal{U}(x)$ in (\ref{condition_b}) is defined uniquely (see Theorem \ref{T:unique} below) .

\item In the case of quantum graphs, another very useful representation of the binding conditions is available, namely splitting into the direct sum of the Dirichlet, Neumann, and Robin type conditions (see part C of Theorem 1.4.4 in \cite{BerKuc}). Regretfully, an analog of such splitting for the quantum open books case, for topological reasons does not necessarily hold continuously with respect to $x\in B$.  
\end{enumerate}
\end{remark}
\bt\label{T:unique}
The matrix $\mathcal{U}(x)$ in (\ref{condition_b}) is defined uniquely.
\et
\begin{proof}
Let us assume that there exist two unitary matrices $\mathcal{U}(x)$ and $\mathcal{V}(x)$ such that they describe the same space of boundary values $\left(u|_B,\dfrac{\partial u}{\partial \nu}|_B\right)$ satisfying 
\begin{equation}\label{$U(x)$}
	i(\mathcal{U}(x)-\mathcal{I})u|_B+(\mathcal{U}(x)+\mathcal{I})\frac{\partial u}{\partial \nu}|_B=0,
\end{equation}
and 
\begin{equation}\label{$V(x)$}
	i(\mathcal{V}(x)-\mathcal{I})u|_B+(\mathcal{V}(x)+\mathcal{I})\frac{\partial u}{\partial \nu}|_B=0.
\end{equation}

We will show that the eigenspaces of the unitary matrices $\mathcal{U}$ and $\mathcal{V}$ are the same for any eigenvalue $\lambda$. This would imply that the matrices are the same.

Thus, let $u$ be an eigenvector of $\mathcal{U}$ corresponding to the eigenvalue $\lambda$.
Let $u$ be such that  $\mathcal{U}(x)u=\lambda u$. We need to show that $u$ belongs to the eigenspace of $\mathcal{V}(x)$  for this  $\lambda$ as well. 
To do that, let us take the following ``boundary values'' pair: 
\be
((\lambda+1)u,-i(\lambda-1)u).
\ee 
By inspection, one sees that it satisfies \eqref{$U(x)$}. So, by assumption, it must also satisfy the condition \eqref{$V(x)$}. Then we get
\begin{align*}
0=& i(\mathcal{V}(x)-\mathcal{I})(\lambda+1)u+(\mathcal{V}(x)+\mathcal{I}) (-i(\lambda-1)u)\\=&2i\mathcal{V}u-2iv,
\end{align*}
which implies
$\mathcal{V}(x) u=\lambda u,$ and thus $u$ is an eigenvector of $\mathcal{V}(x)$  for the eigenvalue  $\lambda$. Reversing the roles of $\mathcal{U}$ and $\mathcal{V}$, we will obtain the opposite inclusion and thus the coincidence of the eigenspaces.  
 Since the eigenspaces of the unitary matrices $\mathcal{U}(x)$ and $\mathcal{V}(x)$ are the same for any eigenvalue, the matrices coincide. This proves the uniqueness of  $\mathcal{U}(x)$.
\end{proof}

\section{Comments and acknowledgments}\label{Comments}
\begin{enumerate}
\item Analyticity of the ``dispersion relation.'' When the binding conditions (i.e., matrices $(A_i(x),B_i(x))$) move through an appropriate manifold of objects, one is interested in how the spectrum of operator changes. Even in the quantum graph case, when one approaches some special (i.e., Dirichlet) types of the vertex conditions, some eigenvalues can disappear at infinity, being otherwise analytic (see, e.g. \cite{BerKuc}). However, it was shown in \cite{KuchZhao} that if the vertex conditions are represented as points of an appropriate Grassmanian $G$, then the ``dispersion relation'', i.e. the graph of the spectrum as a function of the point of $G$ is an analytic set.
    
    One would like to have a similar result for the open books, at least of the limited class that is considered here and $n=2$. 
\item It is definitely very interesting to treat the case of the whole ladder of strata of dimension from zero to $n$, at least for $n=2$. Our assumption of presence of only strata of dimension $n$ and $n-1$ is non-generic and very restrictive. However, removing it seems to be a significant technical challenge, probably requiring one dealing with PDEs in non-smooth domains ``on steroids.'' The same applies for justification of the model in the spirit of \cite{CorbinKuch}.
\end{enumerate}
The work of the first author was supported by Scientific and Technological Research Council of T\"{u}rkiye(T\"{U}B\.{I}TAK) - 2219 International Postdoctoral Research Fellowship Programme and by the Texas A\&M University.
The work of the second author was supported by the NSF DMS Grant 2007408.
\bibliography{Qopenbook}
\bibliographystyle{plain}
\end{document}